\documentclass[11pt,a4paper]{article}%
\usepackage{amsfonts}
\usepackage{amsmath}
\usepackage{amssymb}
\usepackage{graphicx}
\usepackage{geometry}%
\providecommand{\U}[1]{\protect\rule{.1in}{.1in}}
\newtheorem{theorem}{Theorem}

\newtheorem{corollary}{Corollary}
\newtheorem{criterion}{Criterion}
\newtheorem{definition}{Definition}

\newtheorem{lemma}{Lemma}

\newtheorem{proposition}{Proposition}
\newtheorem{remark}{Remark}

\newenvironment{proof}[1][Proof]{\textbf{#1.} }{\  \rule{0.5em}{0.5em}}
\makeatletter
\def \@removefromreset#1#2{\let \@tempb \@elt
\def \@tempa#1{@&#1}\expandafter \let \csname @*#1*\endcsname \@tempa
\def \@elt##1{\expandafter \ifx \csname @*##1*\endcsname \@tempa \else
\noexpand \@elt{##1}\fi}     \expandafter \edef \csname cl@#2\endcsname{\csname cl@#2\endcsname}     \let \@elt \@tempb
\expandafter \let \csname @*#1*\endcsname \@undefined}

\@removefromreset{equation}{section}

\@removefromreset{theorem}{section}
\makeatother
\begin{document}

\title{Specifying nonlocality of a pure bipartite state and analytical relations
between measures for bipartite nonlocality and entanglement }
\author{Elena R. Loubenets$^{1,2}$ and Min Namkung$^{1,3}$\\$^{1}$National Research University Higher School of Economics, \\Moscow 101000, Russia \\$^{2}$ Steklov Mathematical Institute of Russian Academy of Sciences,\\Moscow 119991, Russia\\$^{3}$Department of Applied Mathematics and Institute of Natural Sciences,\\Kyung Hee University, Yongin 17104, Republic of Korea}
\maketitle

\begin{abstract}
For a multipartite quantum state, the maximal violation of all Bell
inequalities constitutes a measure of its nonlocality [Loubenets, \emph{J.
Math. Phys}. 53, 022201 (2012)]. In the present article, for the maximal
violation of Bell inequalities by a pure bipartite state, possibly
infinite-dimensional, we derive \emph{a new upper bound} \emph{expressed in
terms of the Schmidt coefficients of this state. }This new upper bound allows
us also to specify \emph{general analytical relations }between the maximal
violation of Bell inequalities by a bipartite quantum state, pure or mixed,
and such entanglement measures for this state as "negativity" and
"concurrence". To our knowledge, no any general analytical relations between
measures for bipartite nonlocality and entanglement have been reported in the
literature though, for a general bipartite state, specifically such relations
are important for the entanglement certification and quantification scenarios.
As an example, we apply our new results to finding upper bounds on nonlocality
of bipartite coherent states intensively discussed last years in the
literature in view of their experimental implementations.

\end{abstract}

\section{Introduction}

Ever since the seminal paper of Bell \cite{bell} quantum violation of Bell
inequalities was analyzed, analytically and numerically, in many papers and is
now used in many quantum information processing tasks. It is well known that
quantum violation of the Clauser-Horne-Shimony-Holt (CHSH) inequality
\cite{j.f.clauser} cannot exceed
\cite{b.s.cirelson,b.s.tsirelson,r.f.werner,scarani} $\sqrt{2}$ for any
bipartite quantum state, possibly infinite-dimensional. It was also recently
proved \cite{original1, original2} that the maximal quantum violation of the
original Bell inequality\footnote{For the original Bell inequality see also
\cite{loubenets11}.} \cite{bell} is equal to $\frac{3}{2}.$

More generally, quantum violation of any (unconditional) \emph{correlation}
bipartite Bell inequality cannot\footnote{This follows from the definition of
the Grothendieck's constant $K_{G}^{(\mathbb{R})}$\ and Theorem 2.1 in
\cite{b.s.tsirelson}.} exceed the real Grothendiek's constant $K_{G}%
^{(\mathbb{R})}\in\lbrack1.676,1.783]$ but this is not already the case for
quantum violation of a bipartite Bell inequality on joint probabilities and,
more generally,\ quantum violation of a Bell inequality of an arbitrary form,
\emph{a general bipartite Bell inequality }\cite{e.r.loubenets1}, and last
years finding upper bounds on violation by a bipartite quantum state of any
general Bell inequality was intensively discussed within different
mathematical
approaches\ \cite{7,8,e.r.loubenets2,e.r.loubenets2-1,e.r.loubenets6,10,11,12,14,15}%
.

At present, the following general analytical results on quantum violation of
bipartite Bell inequalities are known in the literature:\newline(i) for an
arbitrary bipartite state, pure or mixed, on a Hilbert space $\mathcal{H}%
_{1}\otimes\mathcal{H}_{2}$, $\dim\mathcal{H}_{n}=d_{n}$, violation of any
general Bell inequality\emph{ }with $S_{n}$ settings at $n=1,2$ sites cannot
exceed
\begin{equation}
2\min\{d_{1},d_{2},S_{1},S_{2}\}-1 \label{01}%
\end{equation}
-- in case of generalized quantum measurements (see Eq. 64 in
\cite{e.r.loubenets2}) and \cite{e.r.loubenets4,e.r.loubenets5}%
\begin{align}
\min{\Large \{}d^{\frac{1}{2}},\text{ }3{\Large \}},\text{ \ \ \ \ for \ }S
&  =2,\label{02}\\
\min{\Large \{}d^{\frac{S}{2}},\text{ }2\min\{d,S\}-1{\Large \}},\text{
\ \ \ for }S  &  \geq3,\nonumber
\end{align}
-- in case of projective quantum measurements and $S$ settings per
site;\newline(ii) for the two-qudit Greenberger-Horne-Zeilinger (GHZ) state
$\frac{1}{\sqrt{d}}\sum_{j=1}^{d}|j\rangle^{\otimes2},$ violation of any
general Bell inequality\emph{ }with an arbitrary number of measurement
settings at each of sites admits (Theorem 0.3 in \cite{11}) cannot exceed the
bound $Cd/\sqrt{\ln d}$ where $C$ is an unknown constant independent on a
dimension $d$.

From (\ref{01}) it follows that, for every bipartite quantum state, possibly
infinite-dimensional, violation of any general Bell inequality with fixed
numbers $S_{1},S_{2}\geq1$ of settings at each of two sites is upper bounded
by the value%
\begin{equation}
2\min\{S_{1},S_{2}\}-1, \label{03}%
\end{equation}
whereas, for an arbitrary finite-dimensional bipartite state, pure or mixed,
violation of any general Bell inequality is bounded from above by
\begin{equation}
2\min\{d_{1},d_{2}\}-1. \label{04}%
\end{equation}

In the present article, based on the local quasi hidden variable (LqHV)
formalism developed in \cite{e.r.loubenets2,e.r.loubenets2-1,e.r.loubenets6},
for violation of any general Bell inequality by a pure bipartite state we find
\emph{a new upper bound }which is expressed via the Schmidt coefficients of
this pure state and is tighter than the upper bound (\ref{01}) valid for any
state, pure or mixed.

Based on this new general result, we further specify for a general bipartite
quantum state the analytical relations between its maximal violation of Bell
inequalities on one side and "negativity" and "concurrence" of this state from
the other side. To our knowledge, no any general analytical relations between
measures for bipartite nonlocality and entanglement have been reported in the
literature, though, for a general bipartite state, such relations are
specifically important for finding the minimal amount of entanglement via the
collected experimental data on Bell violation -- the goals of the
semi-device-independent scenario and the device-independent scenario for the
entanglement certification and quantification \cite{000,1,2,3}.

As an example, we apply our new results to finding the upper bounds on
nonlocality of bipartite entangled coherent states intensively discussed last
years in the literature in view of their experimental implementations, see
\cite{m.namkung} and references therein.

The article is organized as follows.

In Section 2, we recall the notion of a general Bell inequality and
specify\ (i) the state parameter characterizing nonlocality a quantum state
and (ii) the analytical upper bound (Theorem 1) on the maximal violation of
Bell inequalities by a bipartite quantum state derived in
\cite{e.r.loubenets2}.

In Section 3, based on Theorem 1, we derive for the maximal violation of Bell
inequalities by an arbitrary pure state new upper bounds (Theorem 2, Corollary
1) expressed in terms of the Schmidt coefficients of this state.

In Section 4, we specify the general analytical relations (Proposition 1 and
Theorem 3) between measures for nonlocality and entanglement of a general
bipartite state, pure or mixed.

In Section 5, we apply our new results to finding upper bounds (Proposition 2)
on the maximal violation of Bell inequalities by bipartite coherent states.

In Section 6, we summarize the main results of the present article.

\section{Preliminaries: general Bell inequalities and nonlocality of a quantum
state}

Consider\footnote{On the probabilistic description of a general correlation
scenario, see \cite{e.r.loubenets0}.} a general bipartite correlation scenario
where each of two participants performs $S_{n}\geq1$, $n=1,2,$ different
measurements, indexed by numbers $s_{n}=1,...,S_{n}$ and with outcomes
$\lambda_{n}\in\Lambda_{n}$. We refer to this correlation scenario as
$S_{1}\times S_{2}$-setting and denote by $\mathrm{P}_{(s_{1},s_{{2}})}%
(\cdot)$ the probability distribution of outcomes $(\lambda_{1},\lambda
_{2})\in\Lambda:=\Lambda_{1}\times\Lambda_{2}$ under the joint measurement
specified by a tuple $(s_{1},s_{2})$ of settings where each $n$-th participant
performs a measurement $s_{n}$ at the $n$-th site. The complete probabilistic
description of such an $S_{1}\times S_{2}$-setting correlation scenario is
given by the family
\begin{align}
\mathcal{P}_{S,\Lambda}  &  :=\left\{  \mathrm{P}_{(s_{1},s_{{2}})}\mid
s_{n}=1,...,S_{n},\text{\ \ \ }n=1,2\right\}  ,\label{1}\\
S  &  :=S_{1}\times S_{2},\nonumber
\end{align}
of joint probability distributions. A correlation scenario admits a local
hidden variable (LHV) model\footnote{For the definition of this notion under a
general correlation scenario, see Definition 4 in \cite{e.r.loubenets0}.} if
each of its joint probability distributions $\mathrm{P}_{(s_{1},s_{2})}%
\in\mathcal{P}_{S,\Lambda}$ admits the representation%
\begin{equation}
\mathrm{P}_{(s_{1},s_{2})}\left(  \mathrm{d}\lambda_{1}\times\mathrm{d}%
\lambda_{2}\right)  =%
{\displaystyle\int\limits_{\Omega}}
P_{1,s_{1}}(\mathrm{d}\lambda_{1}|\omega)\cdot P_{2,s_{{2}}}(\mathrm{d}%
\lambda_{2}|\omega)\nu(\mathrm{d}\omega) \label{2}%
\end{equation}
in terms of a unique probability distribution $\nu(\mathrm{d}\omega)$ of some
variables $\omega\in\Omega$ and conditional probability distributions
$P_{n,s_{n}}(\mathrm{\cdot}|\omega),$ referred to as \textquotedblleft
local\textquotedblright\ in the sense that each $P_{n,s_{n}}(\mathrm{\cdot
}|\omega)$ at $n$-th site depends only on a measurement $s_{n}=1,...,S_{n}$ at
an $n$-th site.

Under an $S_{1}\times S_{2}$-setting correlation scenario described by a
family of joint probability distributions (\ref{1}), consider a linear
combination \cite{e.r.loubenets1}
\begin{equation}
\mathcal{B}_{\Phi_{S,\text{ }\Lambda}}(\mathcal{P}_{S,\text{ }\Lambda}%
):=\sum_{s_{1},s_{{2}}}\left\langle \text{ }\phi_{(s_{1},s_{2})}(\lambda
_{1},\lambda_{2})\right\rangle _{\mathrm{P}_{(s_{1},s_{2})}} \label{3}%
\end{equation}
of the mathematical expectations
\begin{equation}
\left\langle \text{ }\phi_{(s_{1},s_{2})}(\lambda_{1},\lambda_{2}%
)\right\rangle _{\mathrm{P}_{(s_{1},s_{2})}}:=\int\limits_{\Lambda}%
\phi_{(s_{1},s_{2})}(\lambda_{1},\lambda_{2})\text{ }\mathrm{P}_{(s_{1}%
,s_{2})}(\mathrm{d}\lambda_{1}\times\mathrm{d}\lambda_{2}) \label{4}%
\end{equation}
of an arbitrary form, specified by a collection%
\begin{equation}
\Phi_{S,\Lambda}=\left\{  \phi_{(s_{{1}},s_{{2}})}:\Lambda\rightarrow
\mathbb{R}\mid s_{n}=1,...,S_{n};\text{ \ \ }n=1,2\right\}  \label{5}%
\end{equation}
of bounded real-valued functions $\phi_{(s_{{1}},s_{{2}})}$ on $\Lambda
=\Lambda_{1}\times\Lambda_{2}.$ Depending on a choice of a bounded function
$\phi_{(s_{{1}},s_{{2}})}$ and types of outcome sets $\Lambda_{n}$, $n=1,2,$
expression (\ref{4}) can constitute either the probability of some observed
event or if $\Lambda_{n}\subset\mathbb{R},$ $n=1,2,$ the mathematical
expectation (mean) of the product of observed outcomes (called in quantum
information as a correlation function) or have a more complicated form.

If an $S_{1}\times S_{2}$-setting correlation scenario (\ref{1}) admits the
LHV modelling in the sense of representation (\ref{2}), then every linear
combination (\ref{3}) of its mathematical expectations (\ref{4}) satisfies the
\textquotedblleft tight\textquotedblright\footnote{Here, the word a "tight"
LHV constraint means that, in a LHV frame, the bounds established by this
constraint cannot be improved. On the difference between the terms an
"extreme" LHV constraint and a "tight" LHV constraint see section 2.1 of
\cite{e.r.loubenets1}.} LHV constraints \cite{e.r.loubenets1,e.r.loubenets2}
\begin{align}
\mathcal{B}_{\Phi_{S,\Lambda}}^{\inf}  &  \leq\mathcal{B}_{\Phi_{S,\Lambda}%
}(\mathcal{P}_{S,\Lambda})|_{lhv}\leq\mathcal{B}_{\Phi_{S,\Lambda}}^{\sup
},\label{6}\\[0.03in]
\left\vert \text{ }\mathcal{B}_{\Phi_{S,\Lambda}}(\mathcal{P}_{S,\Lambda
})\right\vert _{lhv}  &  \leq\mathcal{B}_{\Phi_{S,\Lambda}}^{lhv}%
:=\max\left\{  \left\vert \mathcal{B}_{\Phi_{S,\Lambda}}^{\sup}\right\vert
,\left\vert \mathcal{B}_{\Phi_{S,\Lambda}}^{\inf}\right\vert \right\}
,\nonumber
\end{align}
where constants
\begin{align}
\mathcal{B}_{\Phi_{S,\Lambda}}^{\sup}  &  :=\sup_{\mathcal{P}_{S,\Lambda}%
\in\mathfrak{G}_{S,\Lambda}^{lhv}}\mathcal{B}_{\Phi_{S,\Lambda}}%
(\mathcal{P}_{S,\Lambda})\label{7}\\[0.03in]
&  =\sup_{\lambda_{n}^{(s_{n})}\in\Lambda_{n},\text{ }\forall s_{n},\text{
}n=1,2}\ \sum_{s_{1},s_{{2}}}\phi_{(s_{1},s_{2})}(\lambda_{1}^{(s_{1}%
)},\lambda_{2}^{(s_{2})}),\nonumber\\
\mathcal{B}_{\Phi_{S,\Lambda}}^{\inf}  &  :=\inf_{\mathcal{P}_{S,\Lambda}%
\in\mathfrak{G}_{S,\Lambda}^{lhv}}\mathcal{B}_{\Phi_{S,\Lambda}}%
(\mathcal{P}_{S,\Lambda})\nonumber\\
&  =\inf_{\lambda_{n}^{(s_{n})}\in\Lambda_{n},\text{ }\forall s_{n},\text{
}n=1,2}\ \sum_{s_{1},s_{{2}}}\phi_{(s_{1},s_{2})}(\lambda_{1}^{(s_{1}%
)},\lambda_{2}^{(s_{2})}).\nonumber
\end{align}
\medskip Here, $\mathfrak{G}_{S,\Lambda}^{lhv}$ denotes the set of all
families (\ref{1}) of joint probability distributions describing $S_{1}\times
S_{2}$-setting correlation scenarios with outcomes in $\Lambda=\Lambda
_{1}\times\Lambda_{2}$ admitting the LHV modelling.

Depending on a form of functional (\ref{3}), which is specified by a family
$\Phi_{S,\Lambda}$ of bounded functions (\ref{6}), some of the LHV constraints
in (\ref{6}) can hold for a wider (than LHV) class of correlation scenarios,
some may be simply trivial, i.e. fulfilled under all correlation scenarios.

\begin{definition}
\cite{e.r.loubenets1,e.r.loubenets2}\emph{\ }Each of the tight linear
LHV\ constraints in (\ref{6}) that can be violated under a non-LHV correlation
scenario is referred to as a general Bell inequality.
\end{definition}

Bell inequalities on correlation functions (like the CHSH inequality) and Bell
inequalities on joint probabilities constitute particular classes of general
Bell inequalities.

If, under a bipartite correlation scenario, all joint measurements
$(s_{1},s_{{2}})$ are performed on a quantum state $\rho$ on a Hilbert space
$\mathcal{H}_{1}\otimes\mathcal{H}_{2}$, then each joint probability
distribution $\mathrm{P}_{(s_{1},s_{{2}})}$ in (\ref{1})\ takes the form
\begin{equation}
\mathrm{P}_{(s_{1},s_{{2}})}(\mathrm{d}\lambda_{1}\times\mathrm{d}\lambda
_{2})=\mathrm{tr}[\rho\{\mathrm{M}_{1}^{(s_{1})}(\mathrm{d}\lambda_{1}%
)\otimes\mathrm{M}_{2}^{(s_{{2}})}(\mathrm{d}\lambda_{2})\}], \label{8}%
\end{equation}
where $\mathrm{M}_{n}^{(s_{n})}(\cdot),$ $\mathrm{M}_{n}^{(s_{n})}(\Lambda
_{n})=\mathbb{I}_{\mathcal{H}_{n}},$ is a normalized positive operator-valued
(POV) measure, describing $s_{n}$-th quantum measurement at $n$-th site. For
this correlation scenario, we denote the family (\ref{1}) of joint probability
distributions by%
\begin{equation}
\mathcal{P}_{S,\Lambda}^{(\rho,\mathfrak{m}_{S,\Lambda})}:=\left\{
\mathrm{tr}[\rho\{\mathrm{M}_{1}^{(s_{1})}(\mathrm{d}\lambda_{1}%
)\otimes\mathrm{M}_{2}^{(s_{{2}})}(\mathrm{d}\lambda_{2})\}],\text{ }%
s_{n}=1,...,S_{n},\text{ }n\in1,2\right\}  , \label{9}%
\end{equation}
where
\begin{equation}
\mathfrak{m}_{S,\Lambda}:=\left\{  \mathrm{M}_{n}^{(s_{n})}\mid\text{ }%
s_{n}=1,...,S_{n},\text{ }n\in1,2\right\}  \label{10}%
\end{equation}
is the collection of all local POV measures at two sites, describing this
quantum correlation scenario.

For a quantum $S_{1}\times S_{2}$-setting correlation scenario\ (\ref{9})
performed\emph{ }on a state $\rho$ on $\mathcal{H}_{1}\otimes\mathcal{H}_{2}$,
possibly infinite-dimensional, every linear combination (\ref{3}) of its
mathematical expectations (\ref{4}) satisfies the \textquotedblleft
tight\textquotedblright\footnote{See Eq. (48) and Lemma 3 in
\cite{e.r.loubenets2}.}\ constraints \cite{e.r.loubenets2}:%
\begin{align}
&  \mathcal{B}_{\Phi_{S,\Lambda}}^{\inf}-\frac{\mathrm{\Upsilon}_{S_{1}\times
S_{2}}^{(\rho,\Lambda)}-1}{2}(\mathcal{B}_{\Phi_{S,\Lambda}}^{\sup
}-\mathcal{B}_{\Phi_{S,\Lambda}}^{\inf})\label{11}\\
&  \leq\mathcal{B}_{\Phi_{S,\Lambda}}(\mathcal{P}_{S,\Lambda}^{\rho
,\mathfrak{m}_{S,\Lambda}})\nonumber\\
&  \leq\mathcal{B}_{\Phi_{S,\Lambda}}^{\sup}+\frac{\mathrm{\Upsilon}%
_{S_{1}\times S_{2}}^{(\rho,\Lambda)}-1}{2}(\mathcal{B}_{\Phi_{S,\Lambda}%
}^{\sup}-\mathcal{B}_{\Phi_{S,\Lambda}}^{\inf}),\nonumber
\end{align}
where
\begin{equation}
1\leq\mathrm{\Upsilon}_{S_{1}\times S_{2}}^{(\rho,\Lambda)}:=\sup_{_{\text{
}\mathfrak{m}_{S,\Lambda},\Phi_{S,\Lambda},\mathcal{B}\text{ }_{\Phi
_{S,\Lambda}}^{lhv}\neq0}}\frac{\left\vert \mathcal{B}_{\Phi_{S,\Lambda}%
}(\mathcal{P}_{S,\Lambda}^{\rho,\mathfrak{m}_{S,\Lambda}})\right\vert
}{\mathcal{B}_{\Phi_{S,\Lambda}}^{lhv}} \label{12}%
\end{equation}
is the maximal violation by a state $\rho$\ of all $S_{1}\times S_{2}$-setting
general Bell inequalities with outcomes $(\lambda_{1},\lambda_{2})\in\Lambda$.

Denote by \cite{e.r.loubenets2}
\begin{align}
1  &  \leq\mathrm{\Upsilon}_{S_{1}\times S_{2}}^{(\rho)}:=\sup_{\Lambda
}\mathrm{\Upsilon}_{S_{1}\times S_{2}}^{(\rho,\Lambda)}\label{16}\\
&  =\sup_{_{\text{ }\Lambda,\mathfrak{m}_{_{S,\Lambda}},\Phi_{_{S,\Lambda}%
},\text{ }\mathcal{B}\text{ }_{\Phi_{S,\Lambda}}^{lhv}\neq0}}\frac{\left\vert
\mathcal{B}_{\Phi_{S,\Lambda}}(\mathcal{P}_{S,\Lambda}^{\rho,\mathfrak{m}%
_{S,\Lambda}})\right\vert }{\mathcal{B}_{\Phi_{S,\Lambda}}^{lhv}}\nonumber
\end{align}
the maximal violation by a state $\rho$\ of all $S_{1}\times S_{2}$-setting
general Bell inequalities for any type of outcomes, discrete or continuous, at
each of two sites.

If, for a state $\rho,$\ the maximal violation $\mathrm{\Upsilon}_{S_{1}\times
S_{2}}^{(\rho)}$ is bounded from above by a value independent on numbers
$S_{1},S_{2}\geq1$, then the parameter \cite{e.r.loubenets4,e.r.loubenets6}%
\begin{equation}
1\leq\mathrm{\Upsilon}_{\rho}:=\sup_{S_{1},S_{2}}\mathrm{\Upsilon}%
_{S_{1}\times S_{2}}^{(\rho)} \label{17}%
\end{equation}
constitutes the maximal violation by a state $\rho$\ of all general Bell
inequalities with any numbers of settings at each site and constitutes a
measure for nonlocality of a state $\rho$ under measurements with all possible
numbers of setting at each of sites.

\begin{definition}
A multipartite quantum state is called nonlocal if it violates a Bell inequality.
\end{definition}

Definition 2 and Eq. (\ref{16}) imply \cite{e.r.loubenets2,e.r.loubenets6}.

\begin{criterion}
A bipartite state $\rho$ is local if and only if parameter $\mathrm{\Upsilon
}_{S_{1}\times S_{2}}^{(\rho)}=1$ for all $S_{1},S_{2}\geq1$ and is nonlocal
if and only if $\mathrm{\Upsilon}_{S_{1}\times S_{2}}^{(\rho)}>1$ for some
numbers $S_{1},S_{2}$ of measurement settings at each of sites.
\end{criterion}

Therefore, parameter $\mathrm{\Upsilon}_{S_{1}\times S_{2}}^{(\rho)}$
constitutes a measure for nonlocality of a bipartite quantum state $\rho$
under correlation scenarios with fixed numbers $S_{1},S_{2}$ of settings at
each of two sites while parameter $\mathrm{\Upsilon}_{\rho}$ is a measure for
nonlocality of a state $\rho$ under correlation scenarios with any numbers of
settings at each of two sites.

According to the upper bound (\ref{01}), for a finite-dimensional bipartite
state, pure or mixed, on $\mathcal{H}_{1}\otimes\mathcal{H}_{2}$,
$\dim\mathcal{H}_{n}=d_{n}$,%
\begin{equation}
\mathrm{\Upsilon}_{\rho}\leq2\min\{d_{1},d_{2}\}-1. \label{18}%
\end{equation}

Let $T_{S_{1}\times S_{2}}^{(\rho)}$ be a self-adjoint trace class dilation of
a state $\rho$ on $\mathcal{H}_{1}\otimes\mathcal{H}_{2}$ to the Hilbert space
$\mathcal{H}_{1}^{\otimes S_{1}}\otimes\mathcal{H}_{2}^{\otimes S_{2}}$. By
its definition%
\begin{align}
&  \mathrm{tr}\left[  T_{S_{1}\times S_{2}}^{(\rho)}\left\{  \mathbb{I}%
_{\mathcal{H}_{1}^{\otimes k_{1}}}\otimes X_{1}\otimes\mathbb{I}%
_{\mathcal{H}_{1}^{\otimes(S_{1}-1-k_{1})}}\otimes\mathbb{I}_{\mathcal{H}%
_{2}^{\otimes k_{2}}}\otimes X_{2}\otimes\mathbb{I}_{\mathcal{H}_{2}%
^{\otimes(S_{2}-1-k_{2})}}\right\}  \right] \label{19}\\
&  =\mathrm{tr}\left[  \rho\left\{  X_{1}\otimes X_{2}\right\}  \right]
,\text{ \ \ }k_{n}=0,...,(S_{n}-1),\text{ \ \ }n=1,2,\nonumber
\end{align}
for all bounded operators $X_{n}$ on $\mathcal{H}_{n}$, $n=1,2.$ Clearly,
$T_{1\times1}^{(\rho)}=\rho$, $\mathrm{tr}[T_{S_{1}\times S_{2}}^{(\rho
)}]=1\ $and $\left\Vert T_{S_{1}\times S_{2}}^{(\rho)}\right\Vert _{1}\geq1,$
where $\left\Vert \cdot\right\Vert _{1}$ means the trace norm.

In \cite{e.r.loubenets2,e.r.loubenets8}, we call a self-adjoint trace class
operator $T_{S_{1}\times S_{2}}^{(\rho)}$ as \emph{an}\textrm{ }$S_{1}\times
S_{2}$-\emph{setting source operator for a state} $\rho$ on $\mathcal{H}%
_{1}\otimes\mathcal{H}_{2}.$ As proved\footnote{This follows from Proposition
1 in \cite{e.r.loubenets2} for a general $N$-partite case.} in
\cite{e.r.loubenets2}, for every bipartite state $\rho$ and arbitrary
integers\emph{\ }$S_{1},S_{2}\geq1$, a source operator $T_{S_{1}\times S_{2}%
}^{(\rho)}$ exists.

\begin{remark}
For a separable quantum state, there always exists a positive source operator.
However, for an arbitrary bipartite quantum state, a source operator
$T_{S_{1}\times S_{2}}^{(\rho)}$ on $\mathcal{H}_{1}^{\otimes S_{1}}%
\otimes\mathcal{H}_{2}^{\otimes S_{2}}$ does not need to be either positive
or, more generally, tensor positive. The latter general notion introduced in
\cite{e.r.loubenets2} means that
\begin{equation}
\mathrm{tr}\left[  T_{S_{1}\times S_{2}}^{(\rho)}\left\{  \text{ }A_{1}%
\otimes\cdots\otimes A_{S_{1}}\otimes B_{1}\otimes\cdots\otimes B_{S_{2}%
}\right\}  \right]  \geq0 \label{20}%
\end{equation}
for all positive bounded operators $A_{k},B_{m}$ on $\mathcal{H}_{1}$ and
$\mathcal{H}_{2},$ respectively.
\end{remark}

Theorem 3 in \cite{e.r.loubenets2} and, more precisely, the second line of Eq.
(53) in this theorem, imply the following analytical upper bound on
$\mathrm{\Upsilon}_{S_{1}\times S_{2}}^{(\rho)}.$

\begin{theorem}
For an arbitrary bipartite quantum state $\rho$, possibly infinite
dimensional, and any integers $S_{n}\geq1,$ $n=1,2,$ the maximal violation
$\mathrm{\Upsilon}_{S_{1}\times S_{2}}^{(\rho)}$ by state $\rho$\ of all
$S_{1}\times S_{2}$-setting general Bell inequalities satisfies the relations%
\begin{align}
1  &  \leq\mathrm{\Upsilon}_{S_{1}\times S_{2}}^{(\rho)}\leq\min\left\{
\inf_{T_{S_{1}\times1}^{(\rho)}}\left\Vert T_{S_{1}\times1}^{(\rho
)}\right\Vert _{1},\inf_{T_{1\times S_{2}}^{(\rho)}}\left\Vert T_{1\times
S_{2}}^{(\rho)}\right\Vert _{1}\right\} \label{21}\\
&  \leq\inf_{T_{S_{1}\times S_{2}}^{(\rho)}}\left\Vert T_{S_{1}\times S_{2}%
}^{(\rho)}\right\Vert _{1},\nonumber
\end{align}
where $\left\Vert \cdot\right\Vert _{1}$ is the trace norm and $T_{S_{1}%
\times1}^{(\rho)},$ $T_{1\times S_{2}}^{(\rho)}$, $T_{S_{1}\times S_{2}%
}^{(\rho)}$ are source operators of state $\rho$ on Hilbert spaces
$\mathcal{H}_{1}^{\otimes S_{1}}\otimes\mathcal{H}_{2},$ $\mathcal{H}%
_{1}\otimes\mathcal{H}_{2}^{\otimes S_{2}}$ and $\mathcal{H}_{1}^{\otimes
S_{1}}\otimes\mathcal{H}_{2}^{\otimes S_{2}},$ respectively.
\end{theorem}

In the following Section, based on Theorem 1, we find for a pure bipartite
state $|\psi\rangle\langle\psi|$ the new upper bounds on its maximal Bell
violation $\mathrm{\Upsilon}_{S_{1}\times S_{2}}^{(|\psi\rangle\langle\psi|)}%
$. These new upper bounds are expressed in terms of the Schmidt coefficients
of a pure state and are tighter than the upper bound (\ref{01}) valid for any
bipartite state, pure or mixed.

\section{New upper bounds for a pure state}

Recall that, for any pure bipartite state $|\psi\rangle\langle\psi|$ on
$\mathcal{H}_{1}\otimes\mathcal{H}_{2}$, $\dim\mathcal{H}_{n}=d_{n},$ the
non-zero eigenvalues $0<\lambda_{k}(\psi)\leq1$ of its reduced states on
$\mathcal{H}_{1}$ and $\mathcal{H}_{2}$ coincide and have the same
multiplicity while vector $|\psi\rangle\in\mathcal{H}_{1}\otimes
\mathcal{H}_{2}$ admits the Schmidt decomposition%
\begin{equation}
|\psi\rangle=\sum_{1\leq k\leq r_{sch}^{(\psi)}}\sqrt{\lambda_{k}(\psi)\text{
}}|e_{k}^{(1)}\rangle\otimes|e_{k}^{(2)}\rangle,\text{ \ \ }\sum_{1\leq k\leq
r_{sch}^{(\psi)}}\lambda_{k}(\psi)=1, \label{22}%
\end{equation}
where each eigenvalue of the reduced states is taken in this sum according to
its multiplicity and $|e_{k}^{(n)}\rangle\in\mathcal{H}_{n}$, $n=1,2,$ are the
normalized eigenvectors of the reduced states of a pure state $|\psi
\rangle\langle\psi|$ on $\mathcal{H}_{1}\otimes\mathcal{H}_{2}$. Parameters
$\sqrt{\lambda_{k}(\psi)\text{ }}$and $1\leq r_{sch}^{(\psi)}\leq
d:=\min\{d_{1},d_{2}\}$ are called the Schmidt coefficients and the Schmidt
rank of $|\psi\rangle,$ respectively. For a separable pure bipartite state,
its Schmidt rank equals to $1$.

It is easy to check that the self-adjoint trace class operators
\begin{align}
T_{1\times S_{2}}^{(\psi)}  &  :=\sum_{k,k_{1}}\sqrt{\lambda_{k}(\psi
)\lambda_{k_{1}}(\psi)}|e_{k}^{(1)}\rangle\langle e_{k_{1}}^{(1)}|\otimes
W_{kk_{1}}^{(2,S_{2})},\label{24}\\
T_{S_{1}\times1}^{(\psi)}  &  :=\sum_{k,k_{1}}\sqrt{\lambda_{k}(\psi
)\lambda_{k_{1}}(\psi)}W_{kk_{1}}^{(1_{,}S_{1})}\otimes|e_{k}^{(2)}%
\rangle\langle e_{k_{1}}^{(2)}|,\nonumber
\end{align}
on $\mathcal{H}_{1}\otimes\mathcal{H}_{2}^{\otimes S_{2}}$ and $\mathcal{H}%
_{1}^{\otimes S_{1}}\otimes\mathcal{H}_{2},$ respectively, where
\begin{align}
W_{kk}^{(n,Sn)}  &  :=\left(  |e_{k}^{(n)}\rangle\langle e_{k}^{(n)}|\right)
^{\otimes S_{n}},\label{25}\\
W_{k\neq k_{1}}^{(n,S_{n})}  &  :=\frac{\left(  |e_{k}^{(n)}+e_{k_{1}}%
^{(n)}\rangle\langle e_{k}^{(n)}+e_{k_{1}}^{(n)}|\right)  ^{\otimes S_{n}}%
}{2^{S_{n}+1}}-\frac{\left(  |e_{k}^{(n)}-e_{k_{1}}^{(n)}\rangle\langle
e_{k}^{(n)}-e_{k_{1}}^{(n)}|\right)  ^{\otimes S_{n}}}{2^{S_{n}+1}}\nonumber\\
&  +i\frac{\left(  |e_{k}^{(n)}+ie_{k_{1}}^{(n)}\rangle\langle e_{k}%
^{(n)}+ie_{k_{1}}^{(n)}|\right)  ^{\otimes S_{N}}}{2^{S_{n}+1}}-i\frac{\left(
|e_{k}^{(n)}-ie_{k_{1}}^{(n)}\rangle\langle e_{k}^{(n)}-ie_{k_{1}}%
^{(n)}|\right)  ^{\otimes S_{n}}}{2^{Sn+1}},\nonumber\\
n  &  =1,2,\nonumber
\end{align}
constitute the $1\times S_{2}$-setting and $S_{1}\times1$-setting source
operators of a pure bipartite state $|\psi\rangle\langle\psi|$.

For source operators (\ref{24}), the trace norms admit the bound
\begin{equation}
\left\Vert T_{1\times S_{2}}^{(\psi)}\right\Vert _{1},\left\Vert
T_{S_{1}\times1}^{(\psi)}\right\Vert _{1}\leq2\left(  \sum_{k}\sqrt
{\lambda_{k}(\psi)}\right)  ^{2}-1, \label{26}%
\end{equation}
which does not depend on numbers $S_{1},S_{2}$ of measurement settings at each
of two sites. Note that
\begin{equation}
2\left(  \sum_{k}\sqrt{\lambda_{k}(\psi)}\right)  ^{2}-1\leq2r_{sch}^{(\psi
)}-1. \label{27}%
\end{equation}
In view of relations (\ref{21}),(\ref{26}) and (\ref{27}) and bound
(\ref{01}), we derive the following new result.

\begin{theorem}
For an arbitrary pure bipartite state $|\psi\rangle\in\mathcal{H}_{1}%
\otimes\mathcal{H}_{2},$ the maximal violation $\mathrm{\Upsilon}_{S_{1}\times
S_{2}}^{(|\psi\rangle\langle\psi|)}$ of $S_{1}\times S_{2}$-setting general
Bell inequalities for any number and type of outcomes at each of sites admits
the bound%
\begin{align}
\mathrm{\Upsilon}_{S_{1}\times S_{2}}^{(|\psi\rangle\langle\psi|)}  &
\leq2\min\left\{  \left(  \sum_{1\leq k\leq r_{sch}^{(\psi)}}\sqrt{\lambda
_{k}(\psi)}\right)  ^{2},S_{1},S_{2}\right\}  -1\label{28}\\
&  \leq2\min\left\{  r_{sch}^{(\psi)},\text{ }S_{1},S_{2}\right\}  -1,
\label{29}%
\end{align}
where $\sqrt{\lambda_{k}(\psi)}$ are the Schmidt coefficients and
$r_{sch}^{(\psi)}$ is the Schmidt rank of a pure bipartite state $|\psi
\rangle.$
\end{theorem}

In view of the relation%
\begin{align}
&  2\min\left\{  r_{sch}^{(\psi)},\text{ }S_{1},S_{2}\right\}  -1\label{30}\\
&  \leq2\min\{d_{1},d_{2},S_{1},S_{2}\}-1,\nonumber
\end{align}
the upper bounds (\ref{28}), (\ref{29}) for a pure state are tighter than the
upper bound (\ref{01}) valid for any state, pure or mixed.

Theorem 2 and relation (\ref{17}) imply

\begin{corollary}
For an arbitrary pure bipartite state $|\psi\rangle\langle\psi|$, possibly
infinite dimensional, the maximal violation $\mathrm{\Upsilon}_{|\psi
\rangle\langle\psi|}$ by this state of general Bell inequalities for any type
of outcomes and any number of settings at each of sites admits the bounds%
\begin{align}
1  &  \leq\mathrm{\Upsilon}_{|\psi\rangle\langle\psi|}\leq2\left(  \sum_{1\leq
k\leq r_{sch}^{(\psi)}}\sqrt{\lambda_{k}(\psi)}\right)  ^{2}-1\label{31}\\
&  \leq2r_{sch}^{(\psi)}-1, \label{32}%
\end{align}

\end{corollary}

\section{Analytical relations between nonlocality and entanglement}

Recall shortly the notions of "negativity" and "concurrence" -- well-known
entanglement measures for a bipartite state.

Negativity $\mathcal{N}_{\rho}$ of a bipartite state $\rho$ on a Hilbert space
$\mathcal{H}_{1}\otimes\mathcal{H}_{2}$, $d_{n}=\dim\mathcal{H}_{n},$ is given
by the relation
\begin{equation}
2\mathcal{N}_{\rho}\mathcal{=}\left\Vert \rho_{T_{n}}\right\Vert _{1}-1,
\label{42}%
\end{equation}
where $\rho_{T_{n}},$ $n=1,2,$ is a partial transpose of $\rho$ with respect
to the subsystem described by a Hilbert space $\mathcal{H}_{n}$ and
$\left\Vert \rho_{T_{1}}\right\Vert _{1}=\left\Vert \rho_{T_{2}}\right\Vert
_{1}$. For a pure bipartite state $|\psi\rangle\langle\psi|$ with the Schmidt
decomposition (\ref{22}), we have\footnote{See, for example, in \cite{4}.}%
\begin{equation}
\left\Vert \text{ }\left(  |\psi\rangle\langle\psi|\right)  _{T_{n}%
}\right\Vert _{1}=\left(  \sum_{1\leq k\leq r_{sch}^{(\psi)}}\sqrt{\lambda
_{k}(\psi)}\right)  ^{2} \label{43}%
\end{equation}
and, hence,
\begin{equation}
2\mathcal{N}_{|\psi\rangle\langle\psi|}=\left(  \sum_{1\leq k\leq
r_{sch}^{(\psi)}}\sqrt{\lambda_{k}(\psi)}\right)  ^{2}-1. \label{44}%
\end{equation}

Concurrence $C_{|\psi\rangle\langle\psi|}$ of a pure state $|\psi
\rangle\langle\psi|$ with the Schmidt decomposition (\ref{22}) is defined
by\footnote{For this notion, see (\ref{4}) and references therein.}%
\begin{equation}
C_{|\psi\rangle\langle\psi|}=\sqrt{2\left(  1-\sum_{1\leq k\leq r_{sch}%
^{(\psi)}}\lambda_{k}^{2}(\psi)\right)  }=\sqrt{2\sum_{k\neq m}\lambda
_{k}(\psi)\lambda_{m}(\psi)}.\label{45}%
\end{equation}
If the concurrence of a pure state is normalized to the unity for maximally
entangled quantum states, then it takes the form\footnote{See in section 7 of
\cite{5}.}%
\begin{equation}
C_{|\psi\rangle\langle\psi|}^{(normal)}=\sqrt{\frac{d}{d-1}\left(
1-\sum_{1\leq k\leq r_{sch}^{(\psi)}}\lambda_{k}^{2}(\psi)\right)  }%
=\sqrt{\frac{d}{d-1}\sum_{k\neq m}\lambda_{k}(\psi)\lambda_{m}(\psi
)},\label{45_1}%
\end{equation}
where $d:=\min\{d_{1},d_{2}\}.$

\begin{lemma}
For an arbitrary pure bipartite state $|\psi\rangle\langle\psi|$ with the
Schmidt rank $r_{sch}^{(\psi)},$
\begin{align}
C_{|\psi\rangle\langle\psi|}  &  \geq\sqrt{\frac{8}{r_{sch}^{(\psi)}%
(r_{sch}^{(\psi)}-1)}}\text{ }\mathcal{N}_{|\psi\rangle\langle\psi|}%
\label{46}\\
&  \geq\sqrt{\frac{8}{d(d-1)}}\text{ }\mathcal{N}_{|\psi\rangle\langle\psi
|},\nonumber
\end{align}
where $d:=\min\{d_{1},d_{2}\}.$
\end{lemma}

\begin{proof}
From relation (8) in \cite{4} it follows%
\begin{equation}
2r_{sch}^{(\psi)}(r_{sch}^{(\psi)}-1)\sum_{k\neq m}\lambda_{k}(\psi
)\lambda_{m}(\psi)\geq2\left\{  \left(  \sum_{k}\sqrt{\lambda_{k}(\psi
)}\right)  ^{2}-1\right\}  ^{2}. \label{47}%
\end{equation}
This, Eq. (\ref{45}) and relation $r_{sch}^{(\psi)}\leq d$ imply
\begin{align}
C_{|\psi\rangle\langle\psi|}  &  \geq\sqrt{\frac{2}{r_{sch}^{(\psi)}%
(r_{sch}^{(\psi)}-1)}}\left\{  \left(  \sum_{k}\sqrt{\lambda_{k}(\psi
)}\right)  ^{2}-1\right\} \label{48}\\
&  =\sqrt{\frac{8}{r_{sch}^{(\psi)}(r_{sch}^{(\psi)}-1)}}\mathcal{N}%
_{|\psi\rangle\langle\psi|}\geq\sqrt{\frac{8}{d(d-1)}}\mathcal{N}%
_{|\psi\rangle\langle\psi|}.\nonumber
\end{align}
This proves the statement.
\end{proof}

From Theorem 2 and Eqs. (\ref{44})--(\ref{46}) it follows.

\begin{proposition}
For an arbitrary pure state $|\psi\rangle\langle\psi|$ on $\mathcal{H}%
_{1}\otimes\mathcal{H}_{2},$ $d_{n}:=\dim\mathcal{H}_{n},$ $d=\min
\{d_{1},d_{2}\}<\infty,$ negativity (\ref{44}) and concurrence (\ref{45})
satisfy the relations
\begin{equation}
\mathcal{N}_{|\psi\rangle\langle\psi|}\geq\frac{\mathrm{\Upsilon}%
_{|\psi\rangle\langle\psi|}-1}{4}, \label{49}%
\end{equation}
and%
\begin{equation}
C_{|\psi\rangle\langle\psi|}\geq\frac{\mathrm{\Upsilon}_{|\psi\rangle
\langle\psi|}-1}{\sqrt{2r_{sch}^{(\psi)}(r_{sch}^{(\psi)}-1)}}\geq
\frac{\mathrm{\Upsilon}_{|\psi\rangle\langle\psi|}-1}{\sqrt{2d(d-1)}},
\label{49_1}%
\end{equation}
where $\mathrm{\Upsilon}_{|\psi\rangle\langle\psi|}$ is the maximal violation
by this state of all Bell inequalities defined by\ (\ref{17}).
\end{proposition}

For a general state $\rho,$ pure or mixed, concurrence $C_{\rho}$ and
negativity $\mathcal{N}_{\rho}$ are defined by relations\footnote{See
\cite{4,00} and references therein.}%
\begin{align}
C_{\rho} &  :=\inf_{\{\alpha_{i},\psi_{i}\}}\sum\alpha_{i}C_{|\psi_{i}%
\rangle\langle\psi_{i}|},\label{50}\\
\mathcal{N}_{\rho} &  :=\inf_{\{\alpha_{i},\psi_{i}\}}\sum\alpha
_{i}\mathcal{N}_{|\psi\rangle\langle\psi|},\label{50_1}%
\end{align}
where $\rho=\sum\alpha_{i}|\psi_{i}\rangle\langle\psi_{i}|,$ $\sum\alpha
_{i}=1,$ $\alpha_{i}>0,$ is a possible convex decomposition of a state $\rho$
via pure states.

\begin{theorem}
For an arbitrary finite-dimensional state $\rho$, pure or mixed,\ on a Hilbert
space $\mathcal{H}_{1}\otimes\mathcal{H}_{2},$ $d_{n}:=\dim\mathcal{H}_{n},$
$d=\min\{d_{1},d_{2}\}<\infty,$ concurrence (\ref{50}) and negativity
(\ref{50_1}) satisfy the relations%
\begin{align}
C_{\rho}  &  \geq\frac{\mathrm{\Upsilon}_{\rho}-1}{\sqrt{2d(d-1)}%
},\label{51_1}\\
\mathcal{N}_{\rho}  &  \geq\frac{\mathrm{\Upsilon}_{\rho}-1}{4}, \label{51}%
\end{align}
where $\mathrm{\Upsilon}_{\rho},$ defined by (\ref{17}), is the maximal
violation by this state of all Bell inequalities.
\end{theorem}

\begin{proof}
From definition (\ref{17}) of $\mathrm{\Upsilon}_{\rho}$ and linearity in
$\rho$ of the functional (\ref{3}) in case of quantum probability
distributions (\ref{8}), we have%
\begin{equation}
\mathrm{\Upsilon}_{\rho}\leq\sum\alpha_{i}\mathrm{\Upsilon}_{|\psi_{i}%
\rangle\langle\psi_{i}|} \label{52}%
\end{equation}
for each possible convex decomposition $\rho=\sum\alpha_{i}|\psi_{i}%
\rangle\langle\psi_{i}|,$ $\sum\alpha_{i}=1,$ $\alpha_{i}>0.$ This implies
\begin{equation}
\mathrm{\Upsilon}_{\rho}\leq\inf_{\{\alpha_{i},\psi_{i}\}}\sum\alpha
_{i}\mathrm{\Upsilon}_{|\psi_{i}\rangle\langle\psi_{i}|}. \label{53}%
\end{equation}
Taking into account in (\ref{53}) relations (\ref{49_1}) and (\ref{50}), we
derive%
\begin{align}
\mathrm{\Upsilon}_{\rho}  &  \leq\inf_{\{\alpha_{i},\psi_{i}\}}\sum_{i}%
\alpha_{i}\mathrm{\Upsilon}_{|\psi_{i}\rangle\langle\psi_{i}|}\leq
1+\sqrt{2d(d-1)}\inf_{\{\alpha_{i},\psi_{i}\}}\sum\alpha_{i}C_{|\psi
_{i}\rangle\langle\psi_{i}|}\label{54}\\
&  =1+\sqrt{2d(d-1)}C_{\rho}.\nonumber
\end{align}
Similarly, by using relations (\ref{53}), (\ref{49}) and (\ref{50_1}), we
have
\begin{align}
\mathrm{\Upsilon}_{\rho}  &  \leq\inf_{\{\alpha_{i},\psi_{i}\}}\sum_{i}%
\alpha_{i}\mathrm{\Upsilon}_{|\psi_{i}\rangle\langle\psi_{i}|}\leq
1+4\inf_{\{\alpha_{i},\psi_{i}\}}\sum\alpha_{i}\mathcal{N}_{|\psi
\rangle\langle\psi|}\\
&  =1+4\mathcal{N}_{\rho}.\nonumber
\end{align}
This proves the statement.
\end{proof}

Since $C_{\rho}\leq\sqrt{\frac{2(d-1)}{d}},$ relation (\ref{51}) immediately
implies the upper bound (\ref{01}), derived otherwise in \cite{e.r.loubenets2}.

\begin{remark}
The analytical bounds (\ref{51_1}) and (\ref{51}) for a general bipartite
state are specifically important for finding the minimal amount of
entanglement via the collected experimental data on Bell violation -- the
goals of the semi-device-independent scenario and the device-independent
scenario for the entanglement certification and quantification
\cite{000,1,2,3}.
\end{remark}

\section{Example}

In this section, we apply the new results of Proposition 1 for specifying
upper bounds on the maximal violation of Bell inequalities by
infinite-dimensional bipartite coherent states of the Bell states like forms
\begin{align}
|\Phi_{1}(\alpha)\rangle &  =\frac{|\alpha\rangle\otimes|\alpha\rangle
+|-\alpha\rangle\otimes|-\alpha\rangle}{\sqrt{2(1+{e^{-4\alpha^{2}})}}%
},\label{55}\\
|\Phi_{2}(\alpha)\rangle &  =\frac{|\alpha\rangle\otimes|-\alpha
\rangle+|-\alpha\rangle\otimes|\alpha\rangle}{\sqrt{2(1+{e^{-4\alpha^{2}})}}%
},\nonumber
\end{align}
where
\begin{equation}
|\pm\alpha\rangle={e^{-\frac{\alpha^{2}}{2}}}\sum_{m=0}^{\infty}\frac
{(\pm\alpha)^{m}}{\sqrt{m!}}|m\rangle\label{56}%
\end{equation}
are the normalized binary coherent states with parameter $\alpha>0$ and
$\{|m\rangle,\ m=0,1,...\}$ are the Fock vectors. For $\alpha\rightarrow0$,
each of bipartite coherent states (\ref{55}) tends to the product state
$|0\rangle\otimes|0\rangle$.

For states (\ref{55}), the \emph{nonzero} eigenvalues of their reduced states
are nondegenerate and are equal to (see Appendix)
\begin{align}
\lambda_{\pm}(\Phi_{1}(\alpha))  &  =\lambda_{\pm}(\Phi_{2}(\alpha
))=\frac{\left(  1\pm{e^{-2\alpha^{2}}}\right)  }{2(1+{e^{-4\alpha^{2}})}}%
^{2}\label{57}\\
&  =\frac{1}{2}\pm\frac{{e^{-2\alpha^{2}}}}{1+{e^{-4\alpha^{2}}}}\nonumber
\end{align}
for all $\alpha>0$ and the Schmidt ranks $r_{sch}^{(\Phi_{j})}=2,$ $j=1,2.$

From Eq. (\ref{57}) it follows%
\begin{align}
\sum_{k}\sqrt{\lambda_{k}(\Phi_{j})} &  =\sqrt{\lambda_{+}(\Phi_{j}(\alpha
)}+\sqrt{\lambda_{-}(\Phi_{j}(\alpha)}=\sqrt{\frac{2}{1+{e^{-4\alpha^{2}}}}%
},\label{58}\\
\sum_{k}\lambda_{k}^{2}(\Phi_{j}) &  =\frac{1}{2}+\frac{2{e^{-2\alpha^{2}}}%
}{\left(  1+{e^{-4\alpha^{2}}}\right)  ^{2}},\label{59}%
\end{align}
From equalities (\ref{58}), (\ref{59}) and Eqs. (\ref{42}), (\ref{45}) it
follows:%
\begin{align}
2\mathcal{N}_{|\Phi_{j}\rangle\langle\Phi_{j}|} &  =\left(  \sum_{k}%
\sqrt{\lambda_{k}(\Phi_{j})}\right)  ^{2}-1=\frac{1-{e^{-4\alpha^{2}}}%
}{1+{e^{-4\alpha^{2}}}},\label{61}\\
C_{|\Phi_{j}\rangle\langle\Phi_{j}|} &  =\sqrt{2\left(  1-\sum_{k}\lambda
_{k}^{2}(\Phi_{j})\right)  }=\frac{1-{e^{-4\alpha^{2}}}}{1+{e^{-4\alpha^{2}}}%
}=2\mathcal{N}_{|\Phi_{j}\rangle\langle\Phi_{j}|}.\label{62}%
\end{align}
The latter equality is consistent with Lemma 1 if $d\rightarrow\infty$ and
$r_{sch}^{(\psi)}=2$. Note that, for each of states (\ref{55}), concurrence
$C_{|\Phi_{j}\rangle\langle\Phi_{j}|}$ is an increasing function of a
parameter $\alpha$ tending to $1$ for $\alpha\rightarrow\infty.$

Relations (\ref{61}), (\ref{62}) and Proposition 1 imply.

\begin{proposition}
For each of infinite-dimensional bipartite coherent states $|\Phi_{j}%
(\alpha)\rangle,$ $j=1,2,$ the maximal violation of Bell inequalities
satisfies the relation%
\begin{equation}
1\leq\mathrm{\Upsilon}_{|\Phi_{j}\rangle\langle\Phi_{j}|}\leq\frac
{3-{e^{-4\alpha^{2}}}}{1+{e^{-4\alpha^{2}}}},\text{ \ \ \ }j=1,2, \label{63}%
\end{equation}
for all $\alpha>0$.
\end{proposition}

\section{Conclusion}

In the present article, based on the local quasi hidden variable (LqHV)
formalism developed in \cite{e.r.loubenets2,e.r.loubenets2-1,e.r.loubenets6},
we find a new upper bound (\ref{28}) (Theorem 2) on the maximal violation
$\mathrm{\Upsilon}_{S_{1}\times S_{2}}^{(|\psi\rangle\langle\psi|)}$ of all
$S_{1}\times S_{2}$-setting Bell inequalities by a pure bipartite state
$|\psi\rangle\langle\psi|,$ possibly infinite-dimensional. This upper bound is
expressed in terms of the Schmidt coefficients of this state and numbers
$S_{1}$, $S_{2}$ of measurement settings at each of two sites and implies
(Corollary 1) that the maximal violation $\mathrm{\Upsilon}_{|\psi
\rangle\langle\psi|}$ of all Bell inequalities by a pure bipartite state
$|\psi\rangle\langle\psi|$ cannot exceed the value $\left(  2\left(  \sum
_{k}\sqrt{\lambda_{k}(\psi)}\right)  ^{2}-1\right)  $.

Based on the new results of Theorem 2, we further find for "negativity" and
"concurrence" of a general bipartite state, pure or mixed, the new lower
bounds (Proposition 1 and Theorem 3) expressed via the maximal violation by
this state of all Bell inequalities. To our knowledge, no any general
analytical relations between measures for bipartite nonlocality and
entanglement have been reported in the literature, though, for a general
bipartite state, such relations as the analytical bounds (\ref{51_1}) and
(\ref{51}) are specifically important for finding the minimal amount of
entanglement via the collected experimental data on Bell violation -- the
goals of the semi-device-independent scenario and the device-independent
scenario for the entanglement certification and quantification
\cite{000,1,2,3}.

As an example, we specify (Proposition 2) the upper bound on the maximal
violation of general Bell inequalities by bipartite entangled coherent states
(\ref{55}) intensively discussed last years in the literature in view of their
experimental implementations \cite{m.namkung}.

\section*{Acknowledgement}

The authors acknowledge the valuable comments of the anonymous reviewers. The
study by E.R. Loubenets in Sections 2, 3 of this work was supported by the
Russian Science Foundation under Grant No. 19-11-00086 and performed at the
Steklov Mathematical Institute of Russian Academy of Sciences. The study by
E.R. Loubenets in Section 4 was performed at the National Research University
Higher School of Economics. The study by Min Namkung in Section 5 was
performed at the National Research University Higher School of Economics until
August 2021, and from September 2021 -- at the Kyung Hee University under the
support from the National Research Foundation of Korea (NRF) grant
(NRF2020M3E4A1080088) funded by the Korean government (Ministry of Science and ICT).

\section*{Appendix}

The vectors
\begin{equation}
|u_{1}\rangle:=|\alpha\rangle,\ \ \ |u_{2}\rangle:=\frac{|-\alpha
\rangle-\langle u_{1}|-\alpha\rangle|u_{1}\rangle}{\sqrt{1-|\langle
u_{1}|-\alpha\rangle|^{2}}}, \tag{A1}\label{A1}%
\end{equation}
where vector $|u_{2}\rangle$ is due to the Gram-Schmidt orthonormalization
process between vectors $|\alpha\rangle$ and $|-\alpha\rangle,$ constitute the
orthonormal basis of the linear span of vectors $|\alpha\rangle$ and
$|-\alpha\rangle$. For $\alpha>0$,%
\begin{equation}
|u_{2}\rangle=\frac{|-\alpha\rangle-{e^{-2\alpha^{2}}}|\alpha\rangle}%
{\sqrt{1-{e^{-4\alpha^{2}}}}} \tag{A2}\label{A2}%
\end{equation}
and bipartite coherent states (\ref{55}) admit the following decompositions:%
\begin{align}
|\Phi_{1}(\alpha)\rangle &  =\frac{(1+e^{-4\alpha^{2}})|u_{1}\rangle
\otimes|u_{1}\rangle+e^{-2\alpha^{2}}\sqrt{1-{e^{-4\alpha^{2}}}}|u_{1}%
\rangle\otimes|u_{2}\rangle}{\sqrt{2(1+{e^{-4\alpha^{2}})}}} \tag{A3}%
\label{A3}\\
&  +\frac{e^{-2\alpha^{2}}\sqrt{1-{e^{-4\alpha^{2}}}}|u_{2}\rangle
\otimes|u_{1}\rangle+(1-e^{-4\alpha^{2}})|u_{2}\rangle\otimes|u_{2}\rangle
}{\sqrt{2(1+{e^{-4\alpha^{2}})}}},\nonumber\\
|\Phi_{2}(\alpha)\rangle &  =\frac{2e^{-2\alpha^{2}}|u_{1}\rangle\otimes
|u_{1}\rangle+\sqrt{1-e^{-4\alpha^{2}}}\left(  \text{ }|u_{1}\rangle
\otimes|u_{2}\rangle+|u_{2}\rangle\otimes|u_{1}\rangle\right)  }%
{\sqrt{2(1+{e^{-4\alpha^{2}})}}}.\nonumber
\end{align}
The nonzero eigenvalues of the reduced states of $|\Phi_{j}(\alpha
)\rangle\langle\Phi_{j}(\alpha)|,$ $j=1,2$ can be easily calculated and are
nongenerate and given by
\begin{equation}
\lambda_{\pm}(\Phi_{j}(\alpha))=\frac{\left(  1\pm{e^{-2\alpha^{2}}}\right)
^{2}}{2(1+{e^{-4\alpha^{2}})}},\text{ \ \ }j=1,2. \tag{A4}\label{A4}%
\end{equation}

\end{document}